\documentclass{amsart}

\usepackage{hyperref}
\usepackage{url}
\def\eps{\epsilon}%
\def\tensor{\,\raise2pt\hbox{${}_{\otimes}$}\,}% Tensor
\def\fdg{\,:\,}% ... fuer die gilt ...
\def\ptl{\partial}% Partial
\def\rest#1{\raise-2pt\hbox{${\lfloor_{#1}}$}}% Restringiert
\def\olin#1{\overline{#1}{}}% Oben-quer
\def\tild#1{\widetilde{#1}{}}% Tilde
\def\grad{{\nabla}}% Gradient
\newcommand{\leftexp}[2]{{\vphantom{#2}}^{#1}{#2}}
\def\halb{\frac{1}{2}}% 1/2
\def \a{\alpha}
\def \b {\beta}

\newtheorem{theorem}{Theorem}[section]

\newtheorem{remark}[theorem]{Remark}
\newtheorem{definition}[theorem]{Definition}
\newcommand{\ba}{\begin{array}}
\newcommand{\ea}{\end{array}}
\newcommand{\bea}{\begin{eqnarray}}
\newcommand{\eea}{\end{eqnarray}}
\newcommand{\bee}{\begin{eqnarray*}}
\newcommand{\eee}{\end{eqnarray*}}

\renewcommand{\gg}{{\bf g}}

\renewcommand{\a}{\alpha}
\renewcommand{\b}{\beta}

\renewcommand{\r}{\rho}

\usepackage{color}

\newcommand{\green}[1]{{\color{green}#1}}

\newcounter{mnotecount}[section]

\renewcommand{\themnotecount}{\thesection.\arabic{mnotecount}}

\newcounter{mymnotecount}[section]
\renewcommand{\themymnotecount}{\thesection.\arabic{mymnotecount}}
\newcommand{\mymnote}[1]{\protect{\stepcounter{mymnotecount}}${\raisebox{0.5\baselineskip}[0pt]{\makebox[0pt][c]{\color{green}{\tiny\em$\bullet$\themnotecount}}}}$\marginpar{\raggedright\tiny\em$\!\bullet$\themymnotecount:

\green{#1}}\ignorespaces}

\renewcommand{\mymnote}[1]{}

\begin{document}

%\title{Asymptotics For Axisymmetric Einstein's Equations With Positive Cosmological Constant}
\title{On 3+1 Lorentzian Einstein Manifolds with One Rotational Isometry}

%    author one information
\author{Nishanth Gudapati}
\address{Department of Mathematics, Yale University, 
10 Hillhouse Avenue, New Haven, CT-06511, USA}
\curraddr{}
\email{nishanth.gudapati@yale.edu}
\thanks{(NG) is supported by a Deutsche Forschungsgemeinschaft  Postdoctoral Fellowship GU1513/1-1 to Yale University}

\iffalse
%    author two information
\author{}
\address{}
\curraddr{}
\email{}
\thanks{}
\fi
%    \subjclass is required.
\subjclass[2010]{Primary 83C05, Secondary 53C50, 35L70}

\date{}

\dedicatory{}

%    Abstract is required.
\begin{abstract}
We consider 3+1 rotationally symmetric Lorentzian Einstein spacetime manifolds with $\Lambda >0$ and reduce the equations to 2+1 Einstein equations coupled to `shifted' wave maps. Subsequently, we prove various (explicit) positive mass-energy theorems. No smallness is assumed.  
\end{abstract}
%\emph{[Prepublished-confidential]}
\maketitle
\section{Lorentzian Einstein Manifolds}
Smooth manifolds whose Ricci curvature is proportional to its metric are called Einstein manifolds. The additional structure endowed on such metrics due to the constant scalar curvature has resulted in  deep connections between differential geometry, string theory and algebraic geometry. We refer the reader to \cite{AB87} for a detailed  exposition on these topics. In this work however, we shall be interested only in Einstein manifolds that are Lorentzian (spacetimes) and positively curved. As we shall see later, the motivation to study these is that they model expanding spacetimes that are consistant with the Einstein's equations for general relativity. 
Suppose $(\bar{M}, \bar{g})$ is a smooth 3+1 dimensional globally hyperbolic Lorentzian spacetime that satisfies the following Einstein's equations for general relativity

\begin{align}\label{ee_pos_cosmo}
\bar{R}_{\mu \nu} - \halb  \bar{R}_{\bar{g}} \,  \bar{g}_{\mu \nu} + \Lambda \bar{g}_{\mu \nu}  = 0, \quad\textnormal{on}\quad (\bar{M}, \bar{g}) 
\end{align}
where $\Lambda >0$ and $\bar{R}_{\mu \nu}$ and $\bar{R}$ are the Ricci curvature and scalar curvature of $(\bar{M}, \bar{g})$ respectively. Equivalently, \eqref{ee_pos_cosmo} transforms to 
\begin{align}
\bar{E}_{\mu \nu} + \Lambda \bar{g}_{\mu \nu} =0, \quad\textnormal{on}\quad (\bar{M}, \bar{g})
\intertext{or}
\bar{R}_{\mu \nu} = \Lambda \bar{g}_{\mu \nu}, \quad\textnormal{on}\quad (\bar{M}, \bar{g}) \label{ee_ricci}
\end{align}
where $\bar{E}_{\mu \nu} \fdg = \bar{R}_{\mu \nu} - \halb  \bar{R}_{\bar{g}} \,  \bar{g}_{\mu \nu}  $ is the Einstein tensor.  In this work we shall follow the sign convention
\begin{align}
\bar{g}(n, n) <0,\quad \textnormal{if $n$ is timelike}.
\end{align}
The case $\Lambda =0$ corresponds to the vacuum Einstein's equations. In 1917 Einstein had introduced $\Lambda$ to his equations with the impression that it would account for the opinion then that the universe is static. In view of the experimental evidence that the large-scale universe is expanding, there is now a general consensus that the simplest way to reconcile the geometry of Einstein's equations with the experimental observations is to introduce the cosmological constant $\Lambda >0$. Mathematically, the case $\Lambda >0$ has garnered significant interest in the last few decades. Since solutions to \eqref{ee_pos_cosmo} admit physically relevant solutions (the de Sitter and the Kerr-de Sitter family), a popular problem has been the study of  stability of these solutions with respect to the initial value problem of \eqref{ee_pos_cosmo}, starting from Friedrich's early landmark results \cite{F86, F86_2,F91}  to  recent remarkable results on \emph{nonlinear} perturbation theory of the Kerr-de Sitter family in the slowly rotating case \cite{HV_16} (see also \cite{M05, SD12,VS12} and the references therein).  Furthermore, there has been an interest in studying asymptotics and gravitational radiation (see the series \cite{ABK1,ABK2,ABK3} and \cite{CI16}).

 Important tools to study the stability and asymptotics of the system \eqref{ee_pos_cosmo} are the  notion of mass-energy that is positive-definite and the estimates on its evolution. In this work we shall lay the groundwork for the initial value problem of a special `axisymmetric' subclass of the system \eqref{ee_pos_cosmo} that admits an additional structure that allows us to study these questions in an explicit manner. This additional structure, recasted in a Hamiltonian framework, has already been quite useful in \cite{GM17} in constructing positive-definite and conserved energy functionals for the \emph{linear} version of the black hole stability problem for the  $\Lambda =0$ case for the full sub-extremal range of the angular-momentum.

 In this work, we shall be interested in the spacetimes satisfying \eqref{ee_pos_cosmo} that admit the conformal null infinity, formally defined as
 
\begin{definition}[Penrose]\label{Penrose}
A Lorentzian spacetime $(\bar{M}, \bar{g})$ satisying \eqref{ee_pos_cosmo} admits a conformal null infinity if 
\begin{enumerate}
\item There exists a Lorentzian spacetime $(\mathcal{M}, \gg)$ with a boundary $\mathcal{I}$ such that $\gg = e^{2 \lambda} g$ in $\bar{M}$, where $\lambda$ is a smooth function in $\mathcal{M}$ such that $\lambda =0, \grad_\gg \lambda \neq 0$ on $\mathcal{I}$

\item $\exists$ a global diffeomorphism $M  \to \mathcal{M}\setminus \mathcal{I}$.

\end{enumerate}
\end{definition}
Then $\mathcal{I}$ is defined as the conformal null infinity of $\bar{M}.$ As it is well-known $\mathcal{I},$ if it exists, is spacelike for $\Lambda >0.$ In the later stages of our work we shall adapt the Definition \ref{Penrose} as per the situation. 

\subsection{Axisymmetric Spacetimes}
Suppose $(\bar{M}, \bar{g})$ is a smooth 3+1 dimensional globally hyperbolic Lorentzian spacetime that is smoothly foliated by a family of simply connected (Riemannian) Cauchy hypersurfaces $(\olin{\Sigma}_t, \bar{q}_t), t \in \mathbb{R}$.  The topology of  $(\bar{M}, \bar{g})$ is $\mathbb{R} \times (\olin{\Sigma}_t, \bar{q}_t) $ such that $\{ p \times \mathbb{R},\, p \in \bar{\Sigma}\}$ is timelike and $(\Sigma_t) \times \{ t\}$ is spacelike. 
\begin{definition}
The Cauchy hypersurfaces $(\olin{\Sigma}_t , \bar{q}_t)$  are axisymmetric if the rotation group $SO(2)$  acts as an isometry with closed orbits and nonempty fixed points set (`axis', denoted by $\Gamma$) for any $t \in \mathbb{R}$. 
\end{definition}
We shall denote $ (\bar{\Sigma}_t, \bar{q}_t)$ for some $t \in \mathbb{R}$ as $(\bar{\Sigma}, \bar{q})$.
Let $\phi \in [0, 2\pi)$ be a parameter associated to the isometry group $SO(2)$ such that $(\ptl_\phi)$ is a Killing vector  of  $(\olin{\Sigma}, \bar{q})$. By construction, $\Gamma$ defines the boundary of the 2 dimensional base $\Sigma \fdg = \bar{\Sigma} \setminus SO(2).$ Further define a 2+1 dimensional Lorentzian manifold $(M, g)$ such that $M \fdg= \bar{M} \setminus SO(2) = \Sigma \times \mathbb{R},$ where the action of the group $SO(2)$ on $\bar{M}$ is as mentioned above. 

In a suitably aligned coordinate system, we can express the metric $\bar{g}$ in terms of a metric 
$g$ of $M$ as follows
\begin{align}\label{kk-ansatz}
\bar{g} = g  + e^{2 u} \Phi^2, \quad\quad \textnormal{(Kaluza-Klein ansatz)}
\end{align}
where $\Phi$ is a 1-form such that $\Phi = d\phi + A_\nu dx^\nu, \nu =0, 1, 2$; the quotient metric $g, u$ and $A_\nu$ are independent of $\phi.$ The case $A \equiv 0$ corresponds to the case where $\ptl_\phi$ is hypersurface orthogonal. %to the case where the orbits of the $SO(2)$ action are orthogonal to $\Sigma.$ 
The case with $A \neq 0$ corresponds to the `twist' of the orbits of the  $SO(2)$ action. 

 As we shall see later, our consideration of axial symmetry of $(\bar{\Sigma}_t, \bar{q}_t)$ covers physically relevant cases, but introduces the following 
complication for $(M, g)$:
\begin{remark}
If the norm of the Killing vector $\ptl_{\phi}$ is given by $ \Vert \ptl_\phi \Vert_{\bar{g}}$, note that, by construction, $\Vert \ptl_\phi (p) \Vert_{\bar{g}} =0$ $\forall$ $p$  in the fixed point set of the $SO(2)$ action on $\bar{M}$ (i.e., on the axis $\Gamma$). 
\end{remark}

\noindent Nevertheless, it is the form of the ansatz \eqref{kk-ansatz} that gives us the additional geometric structure in the equations. 

\subsection{Trace-reversed 2+1 Einstein's equations}
With later use in mind, let us consider the 2+1 version of Einstein's equations in the interior of $(M, g)$ with the cosmological constant $\Lambda'$, the source $T$ and a coupling constant $\kappa$:

\begin{align}\label{ee-cosmo2+1}
R_{\mu \nu} -\halb g_{\mu \nu} R_g + \Lambda' g_{\mu \nu} = \kappa T_{\mu \nu}, \quad (M, g) \setminus \Gamma
\end{align}

Let us also represent \eqref{ee-cosmo2+1} in the trace-reversed form, we have after taking the trace
\begin{align}
R_g  = 2 \,(3\Lambda' - \kappa \text{tr}(T_{\mu \nu})
\end{align}
Thus the trace-reversed Einstein's equations are
\begin{align}\label{ee-cosmo-tr}
\bar{R}_{\mu \nu}  =  2\Lambda' g_{\mu \nu} + \kappa T_{\mu \nu} - \kappa\text{tr}(T_{\mu \nu}) g_{\mu \nu}.
\end{align}
Consider the following energy-momentum tensor

\[ T_{\mu \nu} \fdg = \grad_\mu u \grad_\nu u - \halb g_{\mu \nu} \grad^\sigma u \grad_\sigma u - \halb \Lambda g_{\mu \nu}  e^{-2u}  \]
then the trace of $T$
\[ \text{tr}(T_{\mu \nu})= - \halb \grad^\sigma u \grad_\sigma u - \frac{3}{2} \Lambda e^{-2u} \]
Thus, for wave map fields the quantity $T_{\mu \nu} -  \text{tr}(T_{\mu \nu}) g_{\mu \nu}$  on the right of \eqref{ee-cosmo-tr} is
\begin{align}
T_{\mu \nu} -  \text{tr}(T_{\mu \nu}) g_{\mu \nu}
=  \grad_\mu u \grad_\nu u +  \Lambda g_{\mu \nu} e^{-2u}
\end{align} 
Suppose we are coupling the Einstein's equations with a wave map field, then the trace-reversed equations are

\begin{align}\label{2+1wave}
R_{\mu \nu} = 2\Lambda' g_{\mu \nu} +\kappa ( \grad_\mu u \grad_\nu u +  \Lambda g_{\mu \nu} e^{-2u})
\end{align}
for the choice of choice of $\kappa =2,$
\begin{align}\label{kappa=2}
R_{\mu \nu} = 2\Lambda' g_{\mu \nu} + 2 \grad_\mu u \grad_\nu u + 2 \Lambda g_{\mu \nu} e^{-2u}, \quad (M, g) \setminus \Gamma.
\end{align}
%In particular, note the factor of $2$ on the right of \eqref{2+1wave} (and \eqref{ee-cosmo-tr}) in contrast to \eqref{ee_ricci}.
 
\section{The case with $\ptl_\phi$ hypersurface orthogonal }
The case when $\ptl_\phi$ is hypersurface orthogonal, we have $A \equiv 0$. This is also refered to as the polarized case.

\subsection{Polarized Case}
If $A=0$ uniformly, there is no `twist' or `shift' of the Killing vector $\ptl_\phi$. This case is independently interesting as a few well known solutions of \eqref{ee_pos_cosmo} fall under this category. Thus
\begin{align} \label{polar-ansatz}
\bar{g} = g + e^{2u} d\phi^2
\end{align}
where $u$ and $g$ are independent of $\phi.$
For this form of the metric, the Ricci curvature $\bar{R}$ of $(\bar{M}, \bar{g})$ can be projected on the Ricci curvature $R$ of $(M,g)$ as follows

\begin{subequations} \label{Ricci-high-low-polar}
\begin{align}
\bar{R}_{\mu \nu} =& R_{\mu \nu} - \grad_\mu u \grad_\nu u - \grad_{\mu} \grad_{\nu} u  \label{Rmunu-polar}\\
%\bar{R}_{\mu z} =& \halb e^{-u} \grad_{\sigma} (e^{3 u} F^\sigma _\mu) \\
\bar{R}_{\phi \phi} =& -e^{2u} (g^{\mu \nu} \grad_\mu \grad_\nu u + g^{\mu \nu} \grad_\mu u \grad_\nu u  ) \label{R33-polar}.
\end{align}
\end{subequations}
$\mu, \nu = 0, 1, 2$.
Before we can proceed, we need the following formulas concerning the conformal transformation $\tilde{g} \fdg = e^{2\psi} g:$

\begin{subequations} 
\begin{align}
\sqrt{-\tilde{g}} =& e^{3 \psi} \sqrt{-g} \\
\square_{\tilde{g}} u =&  \frac{1}{\sqrt{-\tilde{g}}} \ptl_\nu \left( \sqrt{-\tilde{g}}\, \tilde{g}^{\mu \nu}\ptl_\mu u \right) = e^{-2 \psi} \left( \square_g u + g^{\mu \nu} \ptl_\nu \psi \, \ptl_\mu u \right) \label{wave-conformal} \\
\tilde{R}_{\mu \nu} =& R_{\mu \nu} - g_{\mu \nu} \grad^\sigma \grad_\sigma \psi -  \grad_\mu  \grad_\nu \psi + \grad_\mu \psi \grad_\nu \psi - g_{\mu \nu} \grad^\sigma \psi \grad_\sigma \psi. \label{Ricci-conformal}
\end{align}
\end{subequations}
$\mu, \nu, \sigma =0, 1, 2.$

Now consider the equation \eqref{R33-polar} for $\bar{R}_{\phi \phi}$ and use the conformal transformation formula \eqref{wave-conformal} for $\psi =u.$
Consequently, in the interior of $(M, g)$, we have,

\begin{align}\label{conformalu1}
\square_{\tilde{g}} u = e^{-2u} \left( \square_g u + g^{\mu \nu} \ptl_\mu u \ptl_\nu u     \right)
\end{align}
%It can be observed that the right side of \eqref{conformalu1} is the term representing $\bar{R}_{zz}$ (upto a factor).
Now we have from the 3+1 Einstein's equations \eqref{R33-polar}
\begin{align}
-e^{2u} \left( \square_g u + \grad^\nu u \grad_\nu u \right)=\bar{R}_{\phi \phi} = \Lambda \bar{g}_{\phi \phi} = \Lambda \, e^{2u}
\end{align}
which in connection with \eqref{conformalu1} implies
 \begin{align}
-e^{4u} (\square_{\tilde{g}} u) = \Lambda e^{2 u}.
 \end{align}
Thus we have,
\begin{align}\label{waveliouville}
\square_{\tilde{g}} u + \Lambda e^{-2 u} =0.
\end{align}

\noindent Now, for the equations \eqref{Rmunu-polar} consider the quantity
\begin{align*}
&\bar{R}_{\mu \nu} + e^{-2u} \bar{g}_{\mu \nu} \bar{R}_{\phi \phi} \\
&= R_{\mu \nu} - \grad_\mu u \grad_\nu u - \grad _\mu \grad_\nu u - g_{\mu \nu} (g^{\a\b} \grad_\a \grad_\b u + g^{\a\b} \grad_\a u \grad_\b u )
\end{align*}
Also the conformally transformed Ricci tensor $\tilde{R}_{\mu \nu}$ in \eqref{Ricci-conformal} for $\psi = u$ in the interior of $M$
\begin{align*}
\tilde{R}_{\mu \nu} =& R_{\mu \nu} - g_{\mu \nu} \grad^\sigma \grad_\sigma u -  \grad_\mu  \grad_\nu u + \grad_\mu u \grad_\nu u - g_{\mu \nu} \grad^\sigma u \grad_\sigma u. 
\end{align*}
So that we have 
\begin{align*}
 2 \Lambda \bar{g}_{\mu \nu} - \tilde{R}_{\mu \nu}=& \bar{R}_{\mu \nu} + e^{-2u} \bar{g}_{\mu \nu} \bar{R}_{\phi \phi}  - \tilde{R}_{\mu \nu}\\
=& -2 \grad_\nu u \grad_\nu u
\end{align*}
Therefore, we have
\begin{align}
\tilde{R}_{\mu \nu} =  2 \tilde{\grad}_\nu u \tilde{\grad}_\nu u + 2\Lambda e^{-2u} \tilde{g}_{\mu \nu} 
\end{align}
In summary we have arrived at the system
\begin{align}
\tilde{R}_{\mu \nu} =&  2 \tilde{\grad}_\nu u \tilde{\grad}_\nu u + 2\Lambda e^{-2u} \tilde{g}_{\mu \nu} \\
\square_{\tilde{g}} u +& \Lambda e^{2u} = 0. \label{liouville-2}
\end{align}
in the interior of $M.$ In particular, note that we have arrived at the system that is as if it is a 2+1 vacuum spacetime that is coupled to a nonlinear equation \eqref{liouville-2} (i.e., $\Lambda' =0$ in \eqref{2+1wave}) and a coupling constant $\kappa =2.$ In other words, the expansion parameter $\Lambda$ is now `transferred' to the forcing term in \eqref{liouville-2}.
% which is especially advantageous as we can choose a time-coordinate gauge condition to kill-off this exponential nonlinearity in energy expressions in the expanding region. 
\subsection{Variational Formulation}
 
 Consider the Einstein-Hilbert action for $\Lambda >0.$
 
 \begin{align}
 \bar{S}_{\text{EH}} \fdg = \halb \int_{\bar{M}} (\bar{R}_{\bar{g}} - 2\Lambda) \bar{\mu}_{\bar{g}}.
 \end{align}
 
 In view of the Gauss-Kodazzi reduction equations \eqref{Ricci-high-low-polar}, we have
 \begin{align}
 \bar{R}_{\bar{g}} = R_g - 2 (\grad^\nu \grad_\nu u + \grad^\nu u \grad_\nu u),
 \end{align}
 upon conformal transformation $\tilde{g} = e^{2u} g$, the scalar curvature transforms as 
 \begin{align}
 \tilde{R}_{\tilde{g}} = e^{-2u} \left(R_g - 2 \grad^\nu u \grad_\nu u -4 \grad^\nu \grad_\nu u  \right)
 \end{align}
 and  \[ \sqrt{-\bar{g}} = e^{-2u} \sqrt{-\tilde{g}} \]
 Thus 
 \begin{align}
 \bar{S}_{\text{EH}} = & \halb \int_{\bar{M}}   (e^{2u} \tilde{R}_{\tilde{g}} + 2 \grad^\nu \grad_\nu u -2 \Lambda) \bar{\mu}_{\bar{g}} \notag\\
 =& \halb \int_{\bar{M}} (e^{2u} \tilde{R}_{\tilde{g}} + 2 e^{2u} \tilde{\grad}^\nu \tilde{\grad}_\nu u -2e^{2u}\tilde{g}^{\mu \nu} \tilde{\grad}^\nu u \tilde{\grad}_\nu u - 2\Lambda ) \bar{\mu}_{\bar{g}} \notag\\
 =& \halb \int_{\tilde{M}} (\tilde{R}_{\tilde{g}}  -2\tilde{g}^{\mu \nu} \tilde{\grad}^\nu u \tilde{\grad}_\nu u - 2 \Lambda e^{-2u} ) \bar{\mu}_{\tilde{g}} \notag\\
 \tilde{S}_{\text{EH}}=& \halb \int_{\tilde{M}} \left(\tilde{R}_{\tilde{g}}  -\kappa (\tilde{g}^{\mu \nu} \tilde{\grad}^\nu u \tilde{\grad}_\nu u + \Lambda e^{-2u} ) \right) \bar{\mu}_{\tilde{g}}
 \end{align}
where we used the convention that two variational functionals are `equal' if their compactly supported first variations in the interior of $M$ produce the same Euler-Lagrange equations. Thus we can recover the 2+1 Einstein's equations if we perform the variation with respect to $\tilde{g}$ and $u$ respectively. For simplicity in notation, we shall henceforth  denote the metric $\tilde{g}$ by $g$ itself.  As additional verification of our reduction, we shall also verify the equation \eqref{waveliouville} for some well known solutions of \eqref{ee_pos_cosmo}.

\begin{theorem}\label{reduction-notwist}
If $u$ satisfies the equation \eqref{liouville-2}, the Einstein's equations \eqref{ee_pos_cosmo} with axial symmetry and $A\equiv 0$ reduce to `vacuum' (i.e., $\Lambda'=0$) Einstein's equations \eqref{kappa=2} in the interior of  $(M, g)$ with a specified source i.e., the stress-energy tensor of the wave equation \eqref{liouville-2}.
\end{theorem}

Finally, we have a few remarks on the equation
\begin{align}\label{waveliouville1}
\square_g u + \Lambda e^{-2 u} =0.
\end{align}
Recall that the constant $\Lambda$ is closely related to the scalar curvature of $(\bar{M}, \bar{g})$. Thus \eqref{waveliouville1} is reminiscent of the Liouville's equation in differential geometry (cf. Nirenberg Problem \cite{M71})
\begin{align}
\leftexp{(2)}\Delta_0 u + K e^{2u}-1 =0
\end{align}
where $K$ is the Gauss curvature  of a 2-manifold and $\leftexp{(2)}\Delta_0$ is the flat 2-Laplacian. Therefore, let us henceforth refer to the equation \eqref{waveliouville} as the `Liouville wave equation'. As we shall see later, this exponential term stays intact even if we turn on the twist term. In a series of beautiful works \cite{MS14,MSMS16}, the Moser-Trudinger inequality was used to study asymptotic behaviour for nonlinear waves with exponential nonlinearities. %In the following we shall consider two examples of spacetimes that have this polarized behaviour and satisfy \eqref{waveliouville}. 
As we shall see, there exist axisymmetric spacetimes with $\Lambda>0$ that admit black holes and it is not inconceivable that the Moser-Trudinger inequality could of use in the study of the nonlinear stability in \emph{compact} domains of such spacetimes (as considered in \cite{HV_16}). 

In the following we shall discuss some well known solutions of \eqref{ee_pos_cosmo} that admit the ansatz 
\eqref{polar-ansatz}. We shall only focus on the aspects that we need, we refer the reader to \cite{AM11}  and Section 5.2 in \cite{HE73} for instance, for more details, Penrose diagrams and other depictions of these spacetimes.

\subsection*{The de Sitter Spacetime}
The de Sitter spacetime is the simplest solution of Einstein's equations with $\Lambda >0$. Geometrically, it is a Lorentizian submanifold obtained by the level sets of the function 
\[ f = -X_0^2 + X_1^2+ X_2^2 + X_3^2 + X_4^2\]
in the Minkowski space $\mathbb{R}^{4+1}.$ Suppose $f=r_c>0$, the de Sitter metric can be expressed globally as 
\begin{align}\label{dS-global}
\bar{g}_{dS} = -dT^2 + r_c^2 \cosh^2 (T/r_c) d\omega^2_{\mathbb{S}^3}
\end{align}
where $ d\omega^2_{\mathbb{S}^3}$ is the standard metric in the unit 3-sphere
$ \mathbb{S}^3.$ The metric $\bar{g}_{dS}$ satisfies the Einstein's equations 
\eqref{ee_pos_cosmo} for $r_c = \sqrt{\frac{3}{\Lambda}}.$

It is easy to see from the form of the metric in \eqref{dS-global} that the standard rotational Killing vector $\ptl_\phi$ acts globally on the de Sitter spacetime through isometries. The de Sitter spacetime admits a static region:
\begin{align}\label{dS-static}
\bar{g}_{dS} = - \left(1-\frac{\Lambda}{3}r^2 \right) dt^2 + \left(  1-\frac{\Lambda}{3}r^2\right)^{-1} dr^2 + r^2 d\omega_{\mathbb{S}^2}
\end{align}
where it admits a timelike Killing vector. Importantly, we should point out that the upper root $r_c = \sqrt{\frac{3}{\Lambda}}$ of $\Delta_{dS} \fdg= 1-\frac{\Lambda}{3}r^2$ is
the cosmological horizon. 
However, the timelike Killing vector in \eqref{dS-static} does not exist globally, in particular the de Sitter spacetime admits an expanding region:

\begin{align}\label{exp-dS}
\bar{g}_{dS} = -  d \hat{t}^2 + e^{2t/r_c} \left( \sum^{3}_{i =1}(d \hat{x}^i)^2\right)
\end{align}
covered by the coordinates $(\hat{t}, \hat{x}_i), i= 1, 2, 3.$
We shall use the form of the metric in \eqref{exp-dS} in asymptotic expansions to define the notion of asymptotically de Sitter spacetimes. If we recast the metric \eqref{exp-dS} in the Kaluza-Klein ansatz:
\begin{align}
\bar{g}_{dS} = e^{-2u} (-e^{2t/r_c}\r^{2} dt^2 + \r^{2}e^{4t/r_c}(d\r^2 + dz^2)) + e^{2u}d\phi^2
\end{align}

where $u = \log \r + t/r_c$ and the 2+1 metric is 
\begin{align}
g_{dS} = -e^{2t/r_c}\r^{2} dt^2 + \r^{2}e^{4t/r_c}(d\r^2 + dz^2)
\end{align}
and its determinant $\sqrt{-g_{dS}} = \r^3 e^{5t/r_c}.$
Now let us verify that the dS metric satisfies the Liouville wave equation for $\r>0$:
\begin{align}
\square_{g_{dS}} u =& \frac{1}{\sqrt{-g_{dS}}} \ptl_\nu (\sqrt{-g_{dS}} \,g^{\mu \nu}_{dS}\ptl_\mu u) 
= \frac{1}{\sqrt{-g_{dS}}} \ptl_t ( \r^3 e^{5t/r_c} (-\r^2 e^{-2t/r_c}) \frac{1}{r_c})\\
=& -\frac{3}{r_c^2 \r^2 e^{-2t/r_c}}
= - \Lambda e^{-2u}
\end{align}
Thus we have recovered the Liouville wave equation
 \[ \square_{g_{dS}} u + \Lambda e^{-2u} =0. \]
 Likewise we can also verify this for the static region (see below).
\subsection*{The Schwarzschild-de Sitter Spacetime}
A more general solution of Einstein's equations \eqref{ee_pos_cosmo} with $\Lambda >0$ is the Schwarzschild de Sitter metric (to be abbreviated henceforth as SdS). Importantly, the SdS metric admits black holes. 
In static coordinates the Schwarzschild-deSitter metric can be expressed as

\begin{align}\label{SdS_static}
\bar{g}_{SdS} = - \left(1-\frac{2m}{r}+\frac{\Lambda}{3}r^2 \right) dt^2 + \left(  1-\frac{2m}{r}-\frac{\Lambda}{3}r^2\right)^{-1} dr^2 + r^2 d\omega_{\mathbb{S}^2}
\end{align} 
which also admits the Killing vector $\ptl_\phi.$ In this work we shall restrict to the cases where the polynomial 
\begin{align}
\Delta_{SdS} = \Lambda r^3 -3r + 6m =0
\end{align}
admits precisely one negative (real) root and two distinct positive (real) roots. With this condition in static coordinates, it is possible to identify positive roots of $\Delta_{SdS}$ as the horizons of the SdS spacetime: $r_{+}$ and $r_c$ with $r_{+} <r_c.$
The SdS spacetime also admits an expanding region, which can be covered by the following coordinates:
\begin{align}\label{SdS_exp}
\bar{g}_{SdS} = - \left( \frac{1-\frac{m}{2r}e^{-t/r_c}}{{1+\frac{m}{2r}e^{-t/r_c}}} \right) d \hat{t}^2 +\left( 1+ \frac{m}{2r} e^{-t/r_c} \right)^4 e^{2t/r_c} \left(\sum^{i=3}_{i=1} (d\hat{x}^i)^2 \right)
\end{align}
where the radial function is defined as $r \fdg=  \sqrt{\sum (\hat{x}^i)^2}.$ Note that the $\hat{t} = const.$ slices are conformally flat in these coordinates. 
Using a different parametrization for $t$ and $r$ in \eqref{SdS_static}, the expanding region can also be represented using the form of the metric in \eqref{SdS_static}. 
Consider the SdS metric in $(t', r', \theta', \phi')$ coordinates:
\begin{align}\label{SdS_exp-static}
\bar{g}_{SdS} = - \left(1-\frac{2m}{r'}+\frac{\Lambda}{3}r'^2 \right) dt'^2 + \left(  1-\frac{2m}{r'}-\frac{\Lambda}{3}r'^2\right)^{-1} dr'^2 + r'^2 d\omega_{\mathbb{S}^2}
\end{align} 
The SdS metric in \eqref{SdS_exp-static} admits a conformal completion as per the Definition \ref{Penrose}, for $\lambda = \frac{1}{r'} \sqrt{\frac{3}{\Lambda}}.$ In \cite{ABK1}
it was noted that the topology of $\mathcal{I}^+$ is $\mathbb{S}^2 \times \mathbb{R}$ and the ends $t' = - \infty$ and $t' = \infty$ were identified as $\iota^0$ (`spatial infinity')and $\iota^{+}$ (future temporal infinity) respectively. Based on this, the notion of asymptotically Schwarzschild-de Sitter spacetime was defined:

\begin{definition}[Asymptotically SdS Lorentzian Einstein Manifolds \cite{ABK1}]
A Lorentzian spacetime $(\bar{M}, \bar{g})$ satisfying \eqref{ee_pos_cosmo} is (future) asymptotically Schwarzschild de-Sitter if 
\begin{itemize}
\item There exists a Lorentzian spacetime $(\mathcal{M}, \gg)$ with a (future) boundary $\mathcal{I}^{+}$ such that $\gg = e^{2 \lambda} g$ in $\bar{M}$, where $\lambda$ is a smooth function in $\mathcal{M}$ such that $\lambda =0, \grad_\gg \lambda \neq 0$ on $\mathcal{I}^{+}$
\item $\exists$ a global diffeomorphism $M  \to \mathcal{M}\setminus \mathcal{I}^{+}$
\item the topology of $\mathcal{I}^{+}$ is $\mathbb{S}^2 \times \mathbb{R} \simeq \mathbb{S}^3 \setminus \{ p_1, p_2\}.$
\end{itemize}
\end{definition}
\noindent where $p_1$ and $p_2$ represent $\iota^0$ and $\iota^+$ respectively.

\noindent Let us now represent the SdS metric in the Kaluza-Klein form:

\begin{align*}
\bar{g}_{SdS} =& r^{-2} \sin^{-2} \theta \left( -r^2 \sin^2 \theta ( 1-\frac{2m}{r} - \frac{\Lambda}{3}r^2) dt^2 + ( 1-\frac{2m}{r} - \frac{\Lambda}{3}r^2)^{-1} dr^2 + r^2 d\theta^2 \right)\\
& \quad + r^2 \sin^2 \theta d\phi^2
\end{align*}
so that the 2+1 metric $g_{SdS}$ is given by 
\begin{align}
 g_{SdS} = -r^2 \sin^2 \theta ( 1-\frac{2m}{r} - \frac{\Lambda}{3}r^2) dt^2 + ( 1-\frac{2m}{r} - \frac{\Lambda}{3}r^2)^{-1} dr^2 + r^2 d\theta^2 
\end{align}
where $e^{2u} = r^2 \sin^2 \theta$ and the metric determinant $\sqrt{-g_{SdS}} = r^4 \sin^3 \theta.$
Now we shall verify if the SdS metric satisfies the Liouville wave equation for $\theta \in (0, \pi)$ and $r_+<r<r_c$: Consider the wave operator
\begin{align*}
\square_{g_{SdS}} u =& \frac{1}{\sqrt{-g_{SdS}}} \ptl_\nu (\sqrt{-g_{SdS}} \ptl^\nu u) \\
=&\frac{1}{\sqrt{-g_{SdS}}} (\ptl_r (r \sin \theta(1-\frac{2m}{r} - \frac{\Lambda}{3}r^2)) + \ptl_\theta (r \sin\theta \, \ptl_\theta(\log r \sin \theta))) \\
=& \frac{1}{r^4 \sin^2 \theta} (\sin\theta (1-\Lambda r^2) -\sin \theta ) = -\frac{\Lambda}{r^2 \sin^2 \theta}
\end{align*}
Thus we have recovered the Liouville wave equation for SdS
\[ \square_{g_{SdS}} u + \Lambda e^{-2u} =0, \]
which reduces to the dS case if we restrict to $m=0.$

\subsection{A Positive-Definite Geometric Energy}
As we alluded to earlier, a positive definite quantity arising from the geometry of Einstein's equations with $\Lambda >0$ would be useful in the estimates in the Cauchy problem. In the case with a rotational symmetry the structure of the equations gives us a few advantages but also introduces subtleties. As such,  the dimensional reduction to the 2+1 dimensional Einstein's equations coupled to the Liouville wave equation results in a canonical energy from the stress-energy tensor of the source: 

\[ E_{div} \fdg = \int_{\Sigma} T(n, n),  \]
where $n$ is the unit timelike normal of $\Sigma \hookrightarrow M.$
Indeed, this energy from the stress-energy tensor has a `defocussing' behaviour and 
 is manifestly positive-definite for $\Lambda >0$ but it is divergent (see \cite{NG17} where the asymptotically flat case ($\Lambda =0$)  is considered but the issues are essentially the same). To obtain a convergent energy we turn to a Hamiltonian quantity of the original 3+1 Einstein's equations:
\begin{align}
\bar{R}_{\mu \nu} = \Lambda \bar{g}_{\mu \nu},\quad \mu, \nu = 0, 1, 2, 3.
\end{align}
Suppose $(\bar{M}, \bar{g})$ admits a 3+1 decomposition such that a smooth (Riemannian) Cauchy hypersurface embeds into it $\bar{\Sigma} \hookrightarrow M $ with a unit time-like normal $n$. 
The Gauss-Kodazzi equations give the following projection relation

\begin{align}
\bar{E} (n, n) = \halb \left(  \bar{R}_{\bar{q}} + (tr_{\bar{q}}(\bar{K}))^2 - \Vert \bar{K} \Vert_{\bar{q}}^2  \right)
\end{align}
where $\bar{K}$ is second fundamental form of the embedding $\bar{\Sigma} \hookrightarrow \bar{M}$ and $\bar{R}_{\bar{q}}$ is the scalar curvature of $(\bar{\Sigma}, \bar{q}).$
Therefore the Hamiltonian constraint $\bar{E}(n, n) + \Lambda \bar{g} (n,n) =0$ now reads
\begin{align}
\bar{R}_{\bar{q}} + (tr_{\bar{q}}(\bar{K}))^2 - \Vert \bar{K} \Vert_{\bar{q}}^2 + 2 \Lambda \bar{g} (n, n) =0.
\end{align}
Note that $n$ is timelike, so in view of our signature $\bar{g} (n, n) < 0,$ results in positive scalar curvature metrics $(\olin{\Sigma}, \bar{q})$ for a suitable time coordinate gauge condition. %\footnote{ as it is generally held that  the positive scalar curvature metrics are more amenable to the notions of total and quasi-local mass \cite{}}. 
Suppose $(\bar{\Sigma}, \bar{q})$ admits the following two forms of the metric
\begin{itemize}
\item Kaluza-Klein coordinates
 \begin{align} \label{wmmetric}
\bar{q} = e^{-2u} ( e^{2q'} (d\rho^2 + dz^2)) + e^{2u}  d\phi^2
\end{align}
\item Brill coordinates
 \begin{align} \label{Brillmetric}
\bar{q} = e^{\sigma} \left( e^{2q} (d\rho^2 + dz^2) + \rho^2 d\phi^2 \right)
\end{align}
\end{itemize} 
It may be noted that the forms of the metric in \eqref{wmmetric} and \eqref{Brillmetric} are closely related. For this work, we shall choose to use the Brill coordinates.  
We shall assume the following regularity conditions on the axis $\Gamma\, (\r =0),$
\begin{align}
q =0, \quad \ptl_\r \sigma = \ptl_\r q =0 \quad \text{on} \quad \Gamma.
\end{align}
For the form of the metric in \eqref{Brillmetric}, the scalar curvature can be calculated as 
\begin{align}
\bar{R}_q =& - \frac{e^{-2q-\sigma}}{2\r} \left(     \r ( \ptl_z \sigma)^2 + (\ptl_\r \sigma )^2 + 4\r  ( \ptl_\r^2 q + \ptl^2_z q + \ptl^2_\r \sigma + \ptl^2_z \sigma ) + 4 \ptl_\r \sigma \right) \notag\\
=& -e^{-2q-\sigma} \left( \halb \vert \grad_0 \sigma \vert^2 + 2 \Delta_0 \sigma + 2 \Delta_0 q -\frac{2}{\rho} \ptl_\rho q   \right)
\intertext{Likewise, in the form of the metric in \eqref{wmmetric}}
\bar{R}_{\bar{q}} =& -2 e^{-2q' + 2u} \left(  (\ptl_z u )^2 + (\ptl_z q')^2 + (\ptl_\rho u)^2 + \ptl^2_\rho q' \right) \\
=& -2 e^{-2q' + 2u} ( \Delta_0 q' - \frac{1}{\rho}  q' + \vert \grad_0 u \vert ^2 )
\end{align}

where 
\begin{align}
 \Delta_0 u \fdg =& \ptl^2_\r u + \ptl^2_z u + \frac{1}{\rho} \ptl_\rho u  \\
 \vert \grad_0 u \vert^2 \fdg=& (\ptl_\r u)^2 + (\ptl_z u)^2
\end{align}

is the $\Delta_0$ flat-space Laplacian and $\vert \grad_0 \cdot \vert^2$ is the (pointwise) Euclidean norm of the flat space gradient applied on axisymmetric function $u$.  The Hamiltonian constraint now reads:
\begin{align}\label{hc1}
-e^{-2q-\sigma} \left( \halb \vert \grad_0 \sigma \vert^2 + 2 \Delta_0 \sigma + 2 \Delta_0 q -\frac{2}{\rho} \ptl_\rho q   \right) + tr(\bar{K})^2 - \Vert \bar{K} \Vert_{\bar{q}}^2 - 2\Lambda \bar{g} (n, n) =&\, 0
\intertext{and} -2 e^{-2q' + 2u} ( \Delta_0 q' - \frac{1}{\rho}  q' + \vert \grad_0 u \vert ^2 ) + tr(\bar{K})^2 - \Vert \bar{K} \Vert_{\bar{q}}^2 - 2\Lambda \bar{g} (n, n) =&\, 0 \label{hc2}
\end{align}

Positivity of charges of Einstein's equations with $\Lambda >0$ close to de Sitter were studied in \cite{AD82}. Separately, using spinor methods various positive energy theorems were proven in \cite{KT02,ST15}. In this work we shall be interested in explicit tensorial positive definite expressions within the class of axial symmetry.  
 We need the following notions of 
asymptotically flat and asymptotically de Sitter spacetimes. Suppose $(\olin{\Sigma}_{AF}, \bar{q}_{AF})$ is a smooth Riemannian manifold, it is called a (globally) asymptotically flat manifold as per the following definition:
\begin{definition}[Asymptotic flatness]
The Riemannian manifold $(\bar{\Sigma}_{AF}, \bar{q}_{AF}, \bar{K}_{AF})$ is (spatially) asymptotically flat if
\begin{itemize}
\item for some compact set $U$, $\exists$ a diffeomorphism from $\bar{\Sigma} \setminus U \to \mathbb{R}^3\setminus \mathbb{B}_1(0)$\, (`end' or asymptotic region)
 \item $(\bar{\Sigma}_{AF}, \bar{q}_{AF}, \bar{K}_{AF})$ admits the asymptotic expansion
\begin{align}\label{AF}
 \bar{q}_{AF} =& \left( 1+\frac{m_{ADM}}{r} \right) \bar{\delta}_{ij} + \mathcal{O} (r^{-1-\alpha}) \\
 \bar{K}_{AF} =& \mathcal{O} (r^{-2-\a})
 \end{align}
 for $\alpha \geq 1/2$ in the asymptotic region.
\end{itemize}
The quantity $m_{ADM}$ is the ADM mass at infinity given by the flux integral
\begin{align}\label{defadm}
m_{ADM} \fdg= \frac{1}{16\pi} \lim_{r \to \infty} \int_{\mathbb{S}^2(r)} \bar{\delta}^{kl} \left( \ptl_k \bar{q}_{il} - \ptl_i \bar{q}_{kl}  \right) \frac{x^i}{r} \bar{\mu}_{\mathbb{S}^2(r)} 
\end{align}
where $r$ is the radius of a coordinate sphere at infinity. 
\end{definition}

As it is well known the positivity of the ADM mass $m_{ADM}$ was proven in the classic works of \cite{schoen-yau-1,schoen-yau-2,witten-pmt}. Using stronger fall-off conditions than in \eqref{AF}, Brill used the explicit form of the Hamiltonian constraint for time-symmetric axially symmetric (polarized) slices to prove an explicit positive mass theorem. It was an early indication of positivity of mass-energy of vacuum Einstein's equations. 
\begin{remark}
In the context of Einstein's equations with $\Lambda >0,$ since the geometry of the well-known solutions is different in different regions, the structure of the Hamiltonian constraint equation \eqref{hc1}(and \eqref{hc2}), which comprises only of  positive-definite terms and divergence terms\footnote{modulo a suitable time-coordinate gauge condition}, is especially useful to define and interpret the flux terms in various regions. 
\end{remark}
Firstly, let us go through the Brill's original argument as it shall be immediately extended to our context.
\begin{theorem}[Brill \cite{B59}] \label{Brill}
Suppose $(\bar{\Sigma}_{AF}, \bar{q}_{AF})$ is asymptotically flat, axially symmetric Riemannian 3-manifold that is (globally) homeomorphic to $\mathbb{R}^3$  and that it can be represented in $(\r, z, \phi)$ coordinates as follows
\begin{align}
\bar{q}_{AF} = e^{\sigma_{AF}} (e^{2q_{AF}} (d\r^2 + dz^2) + \r^2 d\phi^2)
\end{align}
such that it is asymptotically Schwarzschild: 
\begin{enumerate}
\item $e^{\sigma_{AF}/4} \to (1+ \frac{m_B}{2r})$ (asymptotically Schwarzschild)
\item $\ptl_\r (e^{\sigma_{AF}/4})=\ptl_\r q_{AF} =q_{AF} =0$ on the axis $\r=0$
\item  $q_{AF} = \mathcal{O} (r^{-2})$
\end{enumerate}
 as $r \to \infty$, for a constant $m_B$ and 
that $(\bar{\Sigma})_{AF}, \bar{q}_{AF})$ is time-symmetric, then 
\begin{align}
m_{B} \geq 0,
\end{align}
with the equality occurring only if $\bar{q}_{AF}$ is flat. 
\end{theorem}
\begin{proof}
Consider the Hamiltonian constraint equation for $(\bar{\Sigma}_{AF}, \bar{q}_{AF}, \bar{K}_{AF})$ : 
\begin{align}\label{hcaf}
-e^{-2q-\sigma} \left( \halb \vert \grad_0 \sigma_{AF} \vert^2 + 2 \Delta_0 \sigma_{AF} + 2 \Delta_0 q_{AF} -\frac{2}{\rho} \ptl_\rho q_{AF}   \right) + tr(\bar{K}_{AF})^2 - \Vert \bar{K}_{AF}\Vert_{\bar{q}_{AF}}^2 =\, 0
\end{align}
Now consider the integration of the equation \eqref{hcaf} with $tr(K)=0$, in a coordinate ball $B(r)$ of radius $r$, where $r$ is made arbitrarily large in a limiting sense 

\begin{align}\label{hc-int}
-\int _{\mathbb{B}(r)} (\Delta_0 \sigma_{AF} + \Delta_0 q_{AF} - \frac{1}{\rho} \ptl_\r q_{AF}) \r d\r dz d\phi = \halb \int e^{2q_{AF} + \sigma_{AF}} \left(  \Vert K_{AF} \Vert^2_{q_{AF}} \right) + \halb \vert \grad_0 \sigma_{AF} \vert^2 \r d\r dz d\phi.
\end{align}
Due to the asymptotic conditions on $q_{AF}$ the terms $\int \Delta_0 q$ and $\int \frac{1}{\rho} \ptl_\rho q$ drop out as $r \to \infty$ . Now consider the $\int \Delta_0 \sigma_{AF}$ term, integrating it over the interior $B(r)$ of a large sphere $\mathbb{S}^2 (r) \fdg= \{ \rho^2 + z^2 =r^2 \}$ 
\begin{align}
\lim_{r \to \infty} \int_{B(r)} \Delta_0 \sigma_{AF} = \lim_{r \to \infty} \int _{\mathbb{S}^2(r)} \hat{n} (\grad_0 \sigma_{AF})\, r^2 \sin \theta d\theta d\phi
\end{align}
where $\hat{n}$ is the outward normal $1-$form to $\mathbb{S}^2(r).$

Note that $\ptl_r e^{\sigma_{AF}/4}$ and $\frac{1}{4}\ptl_r \sigma_{AF}$ have the same leading order terms ($r^{-2}$)
\begin{align}
e^{\sigma_{AF}/4} =& 1+ \frac{m_B}{2r} + \mathcal{O} (r^{-2}) \\
\log{e^{\sigma_{AF}/4}} = \sigma/4 =& \log (1+ \frac{m_B}{2r} + \mathcal{O}(r^{-2})) \\
=& \frac{m_B}{2r} + \mathcal{O}(r^{-2})
\intertext{this implies}
\ptl_r \sigma_{AF} =& -\frac{2m_B}{r^2} + \mathcal{O}(r^{-3})
\end{align}
in the asymptotic region. Thus
\begin{align}
 \hat{n} (\grad_0 \sigma_{AF}) =& \ptl_r \sigma_{AF} \notag\\
 =& - \frac{2m_B}{r^2} + \mathcal{O}(r^{-3})
\end{align}

\noindent Therefore, 
\begin{align}
\lim_{r\to \infty}\int_{\mathbb{S}^2(r)} \hat{n} (\grad_0 \sigma_{AF}) r^2 \sin \theta d\theta d\phi =& - 8 \pi^2 m_B.
\end{align}
Now, plugging into the equation \eqref{hc-int}, we get the mass formula

\begin{align} \label{mass-weyl}
m_B =& \frac{1}{16\pi^2} \int_{\mathbb{R}^3} 
  e^{2q_{AF} + \sigma_{AF}}  \Vert K \Vert_{\bar{q}_{AF}}^2 + \halb \vert \grad_0 \sigma_{AF} \vert^2 \r d\r dz d\phi \\
  =& \frac{1}{8\pi} \int_{\mathbb{R}^3} 
   e^{2q_{AF} + \sigma_{AF}}  \Vert K_{AF} \Vert_{\bar{q}_{AF}}^2 + \halb \vert \grad_0 \sigma_{AF} \vert^2 \r d\r dz
\end{align}
If we impose the time-symmetry condition, the mass expression simplifies as
\begin{align}\label{Brillmassfinal}
m_B = \frac{1}{16\pi} \int_{\mathbb{R}^3} \vert \grad_0 \sigma_{AF} \vert^2 \r d\r dz \quad \quad \text{(Brill mass formula).}
\end{align}
The rigidity for the equality case $(m=0)$ was also proven using the explicit nature of the expression above. The convergence of the integral in \eqref{Brillmassfinal} follows from the regularity on this axis and the fall-off conditions. Furthermore, the argument allows a straightforward extension to the case with multiple ends. 
\end{proof}

Brill's original work was later extended by Gibbons-Holzegel \cite{GH06} (for 3+1 and 4+1 dimensions with apparent horizons and $A \neq 0$), Dain \cite{D09} (maximal case and $A \neq 0$, see also the nice extension to the linear stability problem for extremal Kerr \cite{DA_14}) and finally by Chru\'sciel \cite{C08}  for the general positively curved axially symmetric case, where the existence of the Brill coordinates for these metrics was also proved. %A short computation for $e^\frac{\sigma_{Sch}}{4} = 1+ \frac{(m_{B})_{Sch}}{2r}$ can be used to verify the formula \eqref{Brillmassfinal} %for the case of Schwarzschild. 
 A formal verification of the correspondence of $m_B$ with $m_{ADM}$, also with weaker fall-off conditions for $q$, follows from Section 3 in \cite{C08}. 

We would like to remark that the ADM mass is captured by the inverse first order fall-off of the asymptotically  flat metric. Since the form of the Hamiltonian constraint has only undifferentiated metric modification from the vacuum case, for the purposes
of the current work we shall adopt the definition of Luo-Xie-Zhang\cite{LXZ10} to define notion of the asymptotically de Sitter metrics. 

\begin{definition}[Asymptotically de Sitter]
A smooth initial data hypersurface $(\bar{\Sigma}, \bar{q})$  that is (globally) homeomorphic to $\mathbb{R}^3$ is called $\Omega-$asymptotically de Sitter if
\begin{itemize}
\item   for a compact set $U$, $\exists$ a diffeomorphism $\bar{\Sigma}\setminus U \to \mathbb{R}^3 \setminus B_1(0)$ (`end' or asymptotic region)
\item $\exists$ a constant $\Omega >0$ such that 
\begin{align}
\bar{q} = \Omega^2 \bar{q}_{AF}, \quad \bar{K} = \Omega \bar{K}_{AF} 
\end{align}
\end{itemize}
in the asymptotic region, where $ \bar{q}_{AF}$ and $\bar{K}_{AF}$ are are such that $(\bar{\Sigma}_{AF}, \bar{q}_{AF}, \bar{K}_{AF})$ is a smooth asymptotically flat metric with convergent ADM mass. 

\end{definition}
Furthermore, in \cite{LXZ10} they prove a positive mass theorem for asymptotically de Sitter metrics based on the positivity of the ADM mass and the following definition of the mass: 

\begin{align}
m (\bar{q}, \bar{K}) \fdg = \Omega\, m_{ADM} (\bar{q}_{AF}, \bar{K}_{AF}).
\end{align}
for CMC slices that satisfy the condition $ \vert tr(K) \vert \leq \sqrt{3 \Lambda}.$ The existence of such slices is also discussed. 
In the following we shall apply this for the case of axial symmetry (in the next section we shall also deal with the  twist $\neq 0$ case). Let us define,
\begin{align}
m (\bar{q}, \bar{K}) \fdg = \Omega\, m_B(\bar{q}_{AF}, \bar{K}_{AF}).
\end{align}
where $m_{B} (\bar{q}_{AF}, \bar{K}_{AF})$ is the Brill mass of the axially symmetric $(\bar{\Sigma}_{AF}, \bar{q}_{AF}, \bar{K}_{AF})$.
\begin{theorem}\label{ADS-mass}
Suppose $(\bar{\Sigma}, \bar{q})$ is an axially symmetric, $\Omega-$asymptotically de Sitter CMC initial data set that is (globally) homeomorphic to $\mathbb{R}^3$ and  that $\vert tr(K) \vert = \sqrt{2 \Lambda}$, $(\bar{\Sigma}_{AF}, \bar{q}_{AF})$ is axially symmetric and asymptotically Schwarzschild in the sense of Theorem \ref{Brill}. Further suppose that $\bar{q}$ can be represented in $(\r, z, \phi)$ coordinates as 
\begin{align} \label{Brillmetric1}
\bar{q} = e^{\sigma} \left( e^{2q} (d\rho^2 + dz^2) + \rho^2 d\phi^2 \right)
\end{align}

(or equivalently  as \eqref{wmmetric}), then the mass $m (\bar{q}, \bar{K})$ converges and 

\begin{align}
m (\bar{q}, \bar{K}) > 0.
\end{align}
 
\end{theorem}

\begin{proof}
Recall the Hamiltonian constraint for the Brill form of the metric:
\begin{align}
-e^{-2q-\sigma} \left( \halb \vert \grad_0 \sigma \vert^2 + 2 \Delta_0 \sigma + 2 \Delta_0 q -\frac{2}{\rho} \ptl_\rho q   \right) + tr(\bar{K})^2 - \Vert \bar{K} \Vert_{\bar{q}}^2 - 2\Lambda \bar{g} (n, n) =&\, 0
\end{align}
Separating out the divergence terms and integrating this like before over a large coordinate ball $\mathbb{B}_r$ of radius $r$, we get
\begin{align}
&\int_{\mathbb{R}^3}(\Delta_0 \sigma + \Delta_0 q - \frac{1}{\r} \ptl_\r q) \r d\r dz \notag\\
&= \halb \int_{\mathbb{R}^3} \left(e^{2q+\sigma} (\Vert \bar{K}\Vert^2_{\bar{q}} - tr(\bar{K})^2 + 2 \Lambda \bar{g}(n,n)) + \halb \vert \nabla_0 \sigma \vert^2 \right) \r d\r dz
\end{align}
For the asymptotically de Sitter $(\bar{\Sigma}, \bar{q})$, like before the $\int \Delta_0 q$ and $\int \frac{1}{\r}\ptl_\r q $ terms drop out. Recall that $\bar{q} = \Omega^2 \bar{q}_{AF}$ in the asymptotic region, where $\bar{q}_{AF}$ is now asymptotically Schwarzschild. Therefore $ \lim_{r \to \infty} \int_{B(r)} \Delta_0 \sigma = -8\pi^2 m_{B} (\bar{q}_{AF}, \bar{K}_{AF})  $ 
\begin{align} \label{asymds-mass}
m(\bar{q}, \bar{K}) = \Omega\, m_{B} (\bar{q}_{AF}, \bar{K}_{AF}) =  \frac{\Omega}{8 \pi} \int_{\mathbb{R}^3} 
   e^{2q + \sigma}  ( \Vert K \Vert_{\bar{q}}^2) + \halb \vert \grad_0 u \vert^2 \r d\r dz 
\end{align}
The convergence of the integral in \eqref{asymds-mass} follows from the convergence of $m_B$ of $(\bar{\Sigma}_{AF}, \bar{K}_{AF})$ and the regularity conditions on the axis. Note that our time coordinate gauge condition $\vert tr(K) \vert = \sqrt{2\Lambda}$ results in $\bar{R}_{q} >0$ and $m >0.$ In addition, it also kills off the exponential term involving $\Lambda$ in the mass-energy expression. 

\end{proof}
\begin{remark}
By exploiting the `time-distance switch' one can also cover the expanding region with  `stationary' coordinates using a `time-translational' Killing vector. Thus, one can study variants of positive mass theorem using maximal hypersurfaces \'a la Brill. In this work however, we shall only be interested in cases where $n$ is manifestly timelike and $(\olin{\Sigma}, \bar{q})$ a Riemannian 3-manifold. 
\end{remark}
\section{The Case with Twist $ v \neq 0$}
We shall now perform the dimensional reduction of $(\bar{M}, \bar{g})$ satisfying \eqref{ee_pos_cosmo} with the ansatz \eqref{kk-ansatz}.
Recall the  the Einstein's equations for the $\Lambda > 0$ case:

\begin{align}\label{ee-cosmo}
\bar{R}_{\mu \nu} - \halb g_{\mu \nu} \bar{R} + \Lambda g_{\mu \nu} =0 \quad \text{for} \quad \Lambda >0 \quad \text{on} \quad (\bar{M}, \bar{g})
\end{align}
\iffalse
Typical solutions of Einstein’s equations with positive cosmological constant are hyperboloidal level-sets belonging to a larger ambient Lorentzian manifold (e.g., deSitter). 

The equation \eqref{ee-cosmo} implies that $(M, g)$ is uniformly positively curved with scalar curvature $4\Lambda$ and 
\fi
 its trace-reversed form
\begin{align}
R_{\mu \nu} = \Lambda g_{\mu \nu}  \quad \text{on} \quad (\bar{M}, \bar{g})
\intertext{and the Kaluza-Klein ansatz}
\bar{g} = g + e^{2u} \Phi^2,
\end{align}
where $g$ is the metric of $M = \bar{M} \setminus SO(2)$ and $g, u, A$ are independent of $\phi.$
Suppose we define the `curvature' $F$ of $A$ in the chosen coordinate frame as

\begin{align}
 F_{\mu \nu} \fdg = \grad_{\mu} A_\nu- \grad_{\nu} A_\mu,
 \end{align}
 Following \cite{YCB-VM96}, the Ricci curvature $(\bar{M}, \bar{g})$ can be projected onto the Ricci curvature of $(M, g)$ as

\begin{subequations} \label{Ricci-high-low}
\begin{align}
\Lambda \bar{g}_{\mu \nu}=\bar{R}_{\mu \nu} =& R_{\mu \nu} - \grad_\mu u \grad_\nu u - \grad_{\mu} \grad_ {\nu} u - \halb e^{2 u} F_{\mu \sigma} F^\sigma_\nu \label{Rmunu} \\
\Lambda \bar{g}_{\mu \phi}=\bar{R}_{\mu \phi} =& -\halb e^{-u} \grad_{\sigma} (e^{3 u} F^\sigma _\mu) \label{Rmu3} \\
\Lambda \bar{g}_{\phi \phi}=\bar{R}_{3 3} =& -e^{2u} (g^{\mu \nu} \grad_\mu \grad_\nu u + g^{\mu \nu} \grad_\mu u \grad_\nu u -\frac{1}{4} e^{2u} F_{\mu \nu} F^{\mu \nu} ) \label{R33}
\end{align}
\end{subequations}
for $\mu, \nu = 0, 1, 2.$
\iffalse
Before we move on, we shall need the following formulas concerning the conformal transformations of Ricci tensor and the covariant wave operator.
Suppose
\[ \tilde{g} \fdg = e^{2\psi} g \]
We have,
\begin{subequations}
\begin{align}
\sqrt{-\tilde{g}} =& e^{3 \psi} \sqrt{-g} \\
\square_{\tilde{g}} u =&
\frac{1}{\sqrt{-\tilde{g}}} \ptl_\nu \left( \sqrt{-\tilde{g}}\, \tilde{g}^{\mu \nu}\ptl_\mu u \right) \notag\\
=&\frac{1}{e^{3 \phi}\sqrt{-g}} \left(  e^{\phi}\ptl_\nu \left( \sqrt{-g}g^{\mu \nu}
 \ptl_\mu u \right)  + e^{\phi}  \sqrt{-g} g^{\mu\nu} \ptl_\nu \phi \ptl_\nu u\right)\notag \\
=& e^{-2 \psi} \left( \square_g u + g^{\mu \nu} \ptl_\nu \psi \, \ptl_\mu u \right)  \\
\tilde{R}_{\mu \nu} =& R_{\mu \nu} - g_{\mu \nu} \grad^\sigma \grad_\sigma \psi -  \grad_\mu  \grad_\nu \psi + \grad_\mu \psi \grad_\nu \psi - g_{\mu \nu} \grad^\sigma \psi \grad_\sigma \psi. 
\end{align}
\end{subequations}
$\mu, \nu, \sigma =0, 1, 2.$
\fi
Fortuitously,  the equation \eqref{Rmu3} allows a twist potential analogous to the $\Lambda =0$ case due to $\bar{g}^\mu_\phi =0$ for $\mu = 0, 1, 2.$
Therefore, the equation \eqref{Rmu3} can be interpreted as the closure of a 1-form. Define

\[ G \fdg = e^{3 u} \, \leftexp{*}{F}, \]
then \eqref{Rmu3} implies
\[ d\, G =0. \]
This in turn implies there exists a twist potential $v$ such that
\[ G = d\,v. \]

The equation $d\, F =0$ for the twist potential $v$ can be transformed into a `wave maps' equation in the conformally transformed manifold $(\tilde{M}, \tilde{g}), \tilde{g} = e^{2\psi} g$ with $\psi = u$ in the interior of $M.$  
\begin{align} \label{twist_pot}
 \tild{\grad}_{\mu}(e^{-3u} \tild{g}^{\mu \nu} \ptl_\nu v) =0 
\end{align}
where $\tilde{\grad}$ is the (Levi-Civita) covariant derivative with respect to the metric $\tilde{g}.$
This equation constitutes one of the `shifted' wave map\footnote{this terminology is temporary, before we come up with a better wording} system. The other equation is given by 
\begin{align}
\tilde{\grad}^\mu \ptl_\mu u + \halb e^{-4u} \tilde{g}^{\mu \nu} \ptl_\mu v \ptl_\nu v + \Lambda e^{-2u}=& \,0.
\end{align}
Analogous to the Theorem \ref{reduction-notwist}, the quantity $\bar{R}_{\mu \nu} + e^{-2u}g_{\mu \nu} \bar{R}_{\phi \phi}$ gives the 2+1 Einstein's equations in a trace reversed form.  
Therefore, we have obtained the 2+1 shifted Einstein-wave map system
\begin{subequations}\label{shifted-ewm}
\begin{align}
E_{\mu \nu} =& \kappa T_{\mu \nu} \\
\tilde{\grad}^\mu \ptl_\mu u + \halb e^{-4u} \tilde{g}^{\mu \nu} \ptl_\mu v \ptl_\nu v + \Lambda e^{-2u}=& \,0 \\
\tild{\grad}_{\mu}(e^{-3u} \tild{g}^{\mu \nu} \ptl_\nu v) =&\, 0.
\end{align}
\end{subequations}
For simplicity in notation, we shall henceforth denote $\tilde{g}$ by $g$ itself. We would like to remark that a subclass of \eqref{shifted-ewm} with $\Lambda =0$ is the 2+1 Einstein-wave map system
\begin{align}
E_{\mu \nu} =& \kappa T_{\mu \nu} \\
\square_g U^i + \leftexp{(h)}{\Gamma}^i_{jk} g^{\a\b} \ptl_\a U^j \ptl_\b U^j =&0
\end{align}
for $\a, \b = 0, 1, 2$ and $i, j = 1, 2$, where
\[ U \fdg (M, g) \to (H^2, h)\]
is a wave map with the target the hyperbolic 2-plane $(\mathbb{H}^2, h).$
\subsection*{The Kerr-de Sitter Spacetime}
The Kerr-de Sitter family of solutions of \eqref{ee-cosmo} are solutions which represent massive rotating black holes. It has twist $v \neq 0.$ The metric is given as follows:
\begin{align}
\bar{g}_{KdS} =& -\frac{\Delta}{\Sigma} \left(\frac{dt - a\sin^2 \theta d\phi}{1+ \frac{\Lambda}{3}a^2} \right)^2 + \frac{\Sigma}{\Delta} dr^2 + \frac{\Sigma}{1+ \frac{\Lambda}{3} a^2 \cos^2 \theta} d\theta^2  \notag\\
&+ 
\frac{\sin^2 \theta (1+ \frac{\Lambda}{3} a^2 \cos^2\theta)}{\Sigma}  \left(\frac{adt- (r^2 +a^2) d\phi}{1+ \frac{\Lambda}{3} a^2} \right)^2
\end{align}
where 
\begin{subequations}
\begin{align}
\Delta =& r^2 -2mr + a^2 -\frac{\Lambda r^2}{3} (r^2 + a^2) \\
\Sigma=& r^2 + a^2 \cos^2 \theta,
\end{align}
\end{subequations}
$\theta \in [0, \phi],$ $\phi \in [0, 2\pi).$ The parameters $m$ and $a$ are the mass and the (specific) angular-momentum respectively.  For the part of the spacetime covered by this form of the metric, the Kerr-de Sitter family admits two Killing vectors
\[ \{ \ptl_t, \ptl_\phi. \} \]

The function $\Delta$ is a quartic polynomial with four roots. In order to avoid extremality of any kind, we shall restrict to the cases of Kerr-de Sitter spacetimes where $-\Delta$ admits precisely one negative root and three distinct non-negative roots - the black hole horizons $r_\pm$ and the cosmological horizon $r_c$. The upper non-negative root of $\Delta$ is the cosmological horizon $(r_c)$ and the lower two roots are the black hole horizons $(r_{\pm})$. Now let us represent the Kerr-de Sitter metric in the Kaluza-Klein or `Weyl-Papapetrou' form in the region $r \in (r_+, r_c)$, where $\Delta >0$. Firstly, we shall transform $(r, \theta)$ part of $\bar{g}_{KdS}$ in a conformally flat form:
introduce the auxillary radial coordinate $R(r)$ such that it is a regular solution of 
\[ \frac{dr}{dR} = \frac{(\Delta)^{1/2}}{R (1+ \frac{\Lambda}{3} a^2 \cos^2 \theta)^{1/2}}, \quad \text{for} \quad r_+ < r <r_c. \]
%This $R$ is given by the integral
%\begin{align}
%R (r) = c \, \text{Exp}{\int^{r_c}_{r_+} \frac{1}{(-\Delta)^{1/2}} dr}
%\end{align}
If we further introduce the coordinates 
\[ \r = R \sin \theta, \quad z = R \cos \theta, \quad \theta \in [0, \pi],\]
the $(r, \theta)$ 2-metric can be transformed into a conformally flat form as
\begin{align}
\frac{\Sigma}{(\r^2 + z^2)(1+ \frac{\Lambda}{3} a^2 \cos^2 \theta)} (d\r^2 + dz^2).
\end{align}
 Now consider the $(\bar{g}_{KdS})_{tt}, (\bar{g}_{KdS})_{t\phi},(\bar{g}_{KdS})_{\phi\phi} $ parts: we have 
 \begin{align}
 (\bar{g}_{KdS})_{tt} =& \frac{1}{\Sigma (1+ \frac{\Lambda}{3}a^2)^2}\left(-\Delta  + a^2\sin^2 \theta \left(1+ \frac{\Lambda}{3} a^2 \cos^2 \theta \right) \right) \notag\\
 (\bar{g}_{KdS})_{t\phi} =&   \frac{a\sin^2 \theta}{\Sigma (1+ \frac{\Lambda}{3}a^2)^2} \left( \Delta - (1+ \frac{\Lambda}{3} a^2 \cos^2 \theta)(r^2+a^2)\right) \notag \\
 (\bar{g}_{KdS})_{\phi\phi} = &  \frac{\sin^2 \theta}{\Sigma (1+ \frac{\Lambda}{3}a^2)^2} \left( - a^2 \sin^2 \theta \Delta +  (1+ \frac{\Lambda}{3} a^2 \cos^2 \theta)(r^2 + a^2)^2 \right)
 \end{align}
 Consider the quantities
 \begin{align}
 \frac{(\bar{g}_{KdS})_{\phi t}}{  (\bar{g}_{KdS})_{\phi \phi} } = \frac{a ( -2mr - \frac{\Lambda}{3}a^2 (r^2 + a^2)(1+ \cos^2 \theta))}{ - a^2 \sin^2 \theta \Delta +  (1+ \frac{\Lambda}{3} a^2 \cos^2 \theta)(r^2 + a^2)^2}
 \end{align}
 and
 %\begin{align}
%&\frac{(\bar{g}_{KdS})_{t t} (\bar{g}_{KdS})_{\phi \phi} - (\bar{g}_{KdS})^2_{t\phi}  }{(\bar{g}_{KdS})_{\phi \phi} }
%\notag\\
%& = \frac{1}{(1+ \frac{\Lambda}{3}a^2)^2} \left(  \frac{(-\Delta  + a^2\sin^2 \theta \left(1+ \frac{\Lambda}{3} a^2 \cos^2 \theta \right))(- a^2 \sin^2 \theta \Delta +  (1+ \frac{\Lambda}{3} a^2 \cos^2 \theta)(r^2 + a^2)^2) -a^2 \sin^2 \theta(\Delta - (1+ \frac{\Lambda}{3} a^2 \cos^2 \theta)(r^2+a^2))^2}{\sin^2 \theta( - a^2 \sin^2 \theta \Delta +  (1+ \frac{\Lambda}{3} a^2 \cos^2 \theta)(r^2 + a^2)^2 )}\right)
 %\end{align}
 
 \begin{align*}
&( \bar{g}_{KdS})_{t t} (\bar{g}_{KdS})_{\phi \phi} - (\bar{g}_{KdS})^2_{t\phi} \notag\\
 &= \frac{1}{\Sigma^2 (1+ \frac{\Lambda}{3} a^2 \cos^2 \theta)^2}  \left( \sin^2 \theta(-\Delta  + a^2\sin^2 \theta (1+ \frac{\Lambda}{3} a^2 \cos^2 \theta ) ) \cdot
 ( - a^2 \sin^2 \theta \Delta +  (1+ \frac{\Lambda}{3} a^2 \cos^2 \theta)(r^2 + a^2)^2 ) \right) \notag\\
 & \quad - \frac{a^2 \sin^4 \theta}{\Sigma^2 (1+ \frac{\Lambda}{3} a^2 \cos^2 \theta)^2} (  \Delta - (1+ \frac{\Lambda}{3} a^2 \cos^2 \theta)(r^2+a^2))^2.
 \end{align*}
Then the Kerr-de Sitter spacetime can be represented in the Kaluza-Klein or `Weyl-Papapetrou' form as 
follows
\begin{align*}
\bar{g}_{KdS}  =& \frac{e^{-2u}}{\Sigma^2 (1+ \frac{\Lambda}{3} a^2 \cos^2 \theta)^2}    \notag\\
& \cdot \left( \sin^2 \theta(-\Delta  + a^2\sin^2 \theta (1+ \frac{\Lambda}{3} a^2 \cos^2 \theta ) ) \cdot
 ( - a^2 \sin^2 \theta \Delta +  (1+ \frac{\Lambda}{3} a^2 \cos^2 \theta)(r^2 + a^2)^2 ) \right ) dt^2 \notag\\
  & - \frac{ e^{-2u}a^2 \sin^4 \theta}{\Sigma^2 (1+ \frac{\Lambda}{3} a^2 \cos^2 \theta)^2} (  \Delta - (1+ \frac{\Lambda}{3} a^2 \cos^2 \theta)(r^2+a^2))^2 dt^2 \notag\\
&+ e^{-2u}\, \frac{\sin^2 \theta \left( - a^2 \sin^2 \theta \Delta +  (1+ \frac{\Lambda}{3} a^2 \cos^2 \theta)(r^2 + a^2)^2 \right)}{(\r^2 +z^2)(1+ \frac{\Lambda}{3} a^2 \cos^2 \theta) (1+ \frac{\Lambda}{3}a^2)^2}  (d\r^2 + dz^2)\notag\\
& + e^{2u} (d\phi + \frac{a ( -2mr - \frac{\Lambda}{3}a^2 (r^2 + a^2)(1+ \cos^2 \theta))}{ - a^2 \sin^2 \theta \Delta +  (1+ \frac{\Lambda}{3} a^2 \cos^2 \theta)(r^2 + a^2)^2} dt )^2
\end{align*}

where the norm and the vector potential corresponding to the `shift' of the Killing vector $\ptl_\phi$ are given by 
\begin{subequations}
\begin{align}
e^{2u} =& \frac{\sin^2 \theta}{\Sigma (1+ \frac{\Lambda}{3}a^2)^2} \left( - a^2 \sin^2 \theta \Delta +  (1+ \frac{\Lambda}{3} a^2 \cos^2 \theta)(r^2 + a^2)^2 \right) \\
A_{t} =&   \frac{a ( -2mr - \frac{\Lambda}{3}a^2 (r^2 + a^2)(1+ \cos^2 \theta))}{ - a^2 \sin^2 \theta \Delta +  (1+ \frac{\Lambda}{3} a^2 \cos^2 \theta)(r^2 + a^2)^2} \\
\intertext{respectively. Likewise, the conformal factor of the 2-metric in $(\r, z)$ is given by}
e^{2q}\fdg =& \frac{\sin^2 \theta}{(\r^2 +z^2)(1+ \frac{\Lambda}{3} a^2 \cos^2 \theta) (1+ \frac{\Lambda}{3}a^2)^2} \left( - a^2 \sin^2 \theta \Delta +  (1+ \frac{\Lambda}{3} a^2 \cos^2 \theta)(r^2 + a^2)^2 \right)
\end{align}
\end{subequations}

\iffalse
the 2+1 lapse function can be simplified further
\begin{align}
\tilde{N} = 
\end{align}
The formula above are quite technical but we would like to make a few takeaway remarks
\begin{itemize}
\item At the points on the axis $\theta \in \{ 0, \pi \}$, the behaviour of the functions $q$ and $u$ are such that...
\item the behaviour of these functions close to the black hole horizon $r_+$
\end{itemize}
\fi
Note that the norm of the Killing vector, the conformal factor and the lapse function of the $2+1$ metric $g$  vanish on the axis $\theta \in \{ 0, \pi \}$. 
For the special case of $a =0$ we recover the Schwarzschild-de Sitter metric in isotropic coordinates: 
\begin{subequations}
\begin{align}
e^{2u} =&\, r^2 \sin^2 \theta\\
A_{t} =&\,  0 \\
\intertext{the conformal factor of the 2-metric $(\r, z)$}
e^{2q}=& \frac{r^4\sin^2 \theta}{\r^2 + z^2}.
\end{align}
\end{subequations}
\iffalse
but the function $R(r)$ can be solved explicitly:

The lapse function $\tilde{N}^{-2} = \Delta \sin^2 \theta$, thus $\tilde{N} \to 0$ as $r \to r_+$ for SdS.
\fi
% The form of the Schwarzschild-de Sitter metric in isotropic coordinates is
%\begin{align*}
%\bar{g}_{SdS} =  
%\end{align*}

%An equivalent formula can be obtained if we start with \eqref{wmmetric} form of the metric.  In the formula \eqref{mass-weyl} the right hand side can be represented as the energy of a wave equation. 
\noindent Let us now move on to the positive mass expressions for the $v \neq 0$ case.
\subsubsection*{Kaluza-Klein coordinates}
 Consider the Kaluza-Klein ansatz for the metric $\bar{q}$
\begin{align} \label{kk-twist}
\bar{q} =&e^{-2u} (e^{2q'} (d\r^2 + dz^2)) + e^{2u}  (d\phi + A_i dx^i)^2,
\quad i = 1, 2.
\end{align}
The scalar curvature $\bar{R}_{\bar{q}}$ in this parametrization of the metric $\bar{q}$ is
\begin{align}
\bar{R}_{\bar{q}} =& - \halb e^{-4q' + 6u} (\ptl_z A_\rho - \ptl_\rho A_z)^2 -2 e^{-2q' + 2u} ( \Delta_0 q' - \frac{1}{\rho}  q' + \vert \grad_0 u \vert ^2 )
\end{align}
Suppose the 2+1 decomposition of the metric $(M, g)$ can be represented as 
\begin{align}
g = - \tilde{N}^2 dt^2 + \gamma_{ab} (dx^a + \tilde{N}^a dt) (dx^b + \tilde{N}^b dt),
\end{align}
$a = 1, 2$ where $\tilde{N}$ and $\tilde{N}^a$ are the 2+1 lapse and shift respectively. Suppose $\eps^{ab}$ is the Levi-Civita antisymmetric symbol, then the quantity $\eps^{ab} (\ptl_a A_b - \ptl_b A_a)$ for $a, b =1, 2$ can be expressed in terms of the \emph{spacetime} twist 
potential $v$ as  
\begin{align}
\halb \eps^{ab} (\ptl_a A_b - \ptl_b A_a) = \tilde{N} e^{2q'-4u} (\ptl_t v - \tilde{N}^a \ptl_a v).
\end{align}
Therefore, the scalar curvature $\bar{R}_{\bar{q}}$ can be expressed in terms of $v$ as
\begin{align}
\bar{R}_{\bar{q}} =& - \halb \tilde{N}^2 e^{-2u} ( \ptl_t v - \tilde{N}^a \ptl_a v)^2 -2 e^{-2q' + 2u} ( \Delta_0 q' - \frac{1}{\rho}  q' + \vert \grad_0 u \vert ^2 )
\end{align}
Therefore, the Hamiltonian constraint is
\begin{align}
 - \halb \tilde{N}^2 e^{-2u} (\ptl_t v - \tilde{N}^a \ptl_a v)^2 -2 e^{-2q' + 2u} ( \Delta_0 q' - \frac{1}{\rho}  q' + \vert \grad_0 u \vert ^2 ) \notag\\ 
 + tr(K)^2  - \Vert K \Vert^2_{\bar{q}} + 2 \Lambda \bar{g}(n, n) =0.
\end{align}
\subsubsection*{Brill coordinates}
Likewise, if we consider a more general form of the original Brill metric:
\begin{align}\label{brill-twist}
q =& e^{\sigma} (e^{2q} (d\r^2 + dz^2) + \r^2 ( d\phi + A_a dx^a)^2), \quad a, b = 1, 2
\end{align}
where again the $1-$form $A$ is the vector potential corresponding to the `shift' of the Killing vector $\ptl_\phi$, the functions $\sigma, q, A$ are regular functions such that 
their derivatives with respect to $\r$ vanish on the axis $\Gamma$ and $ q=0$ on $\Gamma.$ 
Then the scalar curvature of $(\bar{\Sigma}, \bar{q})$ in these coordinates is given by 
\begin{align}
\bar{R}_{\bar{q}} =& \frac{1}{2} e^{-4q - \sigma} (-\r^2 (\ptl_z A_\r - \ptl_\r A_z)^2 ) -e^{-2q-\sigma} \left( \halb \vert \grad_0 \sigma \vert^2 + 2 \Delta_0 \sigma + 2 \Delta_0 q -\frac{2}{\rho} \ptl_\rho q   \right)  
\end{align}

Therefore the Hamiltonian constraint is
\begin{align}
 \frac{1}{2} e^{-4q + 2\sigma} (-\r^2 (\ptl_z A_\r - \ptl_\r A_z )^2 )  -e^{-2q-\sigma} \left( \halb \vert \grad_0 \sigma \vert^2 + 2 \Delta_0 \sigma + 2 \Delta_0 q -\frac{2}{\rho} \ptl_\rho q   \right)  \notag\\
  + tr(K)^2  - \Vert K \Vert^2_{\bar{q}} + 2 \Lambda \bar{g}(n, n) =0.
\end{align}

In the following we shall prove a positive mass thereom for the case $A\neq0$. As in Theorem \ref{ADS-mass}, the asymptotic behaviour of $(\bar{\Sigma}, \bar{q}, \bar{K})$ is governed by the asymptotics of $(\bar{\Sigma}_{AF}, \bar{q}_{AF}, \bar{K}_{AF})$ with convergent Brill mass. In particular, we shall assume that $\bar{q}_{AF}$ can be represented
in terms of Brill coordinates 
\begin{align}\label{brill-twist-af}
\bar{q}_{AF} =& e^{\sigma_{AF}} (e^{2q_{AF}} (d\r^2 + dz^2) + \r^2 ( d\phi + (A_{AF})_a dx^a)^2), \quad a, b = 1, 2
\end{align}
where $\sigma_{AF}, q_{AF}, A_{AF}$ are regular functions such that
\begin{align}\label{af-asym}
\sigma_{AF} = \mathcal{O}(r^{-1}), \ptl \sigma_{AF} = \mathcal{O}(r^{-2}), q_{AF} = \mathcal{O}(r^{-2}), \ptl q_{AF} = \mathcal{O}(r^{-3})
\intertext{and}
(A_{AF})_\r, (A_{AF})_z = \mathcal{O}(r^{-2}), \quad (\ptl A_{AF})_\r, (\ptl A_{AF})_z = \mathcal{O}(r^{-3})
\end{align}
in the asymptotic region of $(\bar{\Sigma}_{AF}, \bar{q}_{AF}, \bar{K}_{AF})$, $r = \sqrt{\r^2 + z^2}.$ Furthermore, they satisfy the regularity conditions on the
axis
\begin{align}\label{axis-af}
 \ptl_\r \sigma_{AF} = \ptl_\r q_{AF} = q_{AF} =0, \quad \text{on} \quad \r=0.
\end{align} 

In this setting, we shall define 
\begin{align}\label{ADS-def-twist}
 m (\bar{q}, \bar{K}) \fdg = \Omega \,m_B (\bar{q}_{AF}, \bar{K}_{AF}).
 \end{align}
\begin{theorem}
Suppose $(\bar{\Sigma}, \bar{q})$ is axially symmetric CMC initial data hypersurface with $\vert tr(K) \vert = \sqrt{2 \Lambda}$ that is (globally) homeomorphic to $\mathbb{R}^3$ and is $\Omega-$asymptotically de Sitter with $(\bar{\Sigma}_{AF}, \bar{q}_{AF}, \bar{K}_{AF})$ satisfying \eqref{brill-twist-af} $-$ \eqref{axis-af}. Further suppose that the metric
$\bar{q}$ can be represented in the form \eqref{brill-twist} (or equivalently \eqref{kk-twist} ) then the mass $m (\bar{q}, \bar{K})$  as defined in \eqref{ADS-def-twist} converges and 
\begin{align}
m (\bar{q}, \bar{K}) >0.
\end{align}

\end{theorem}
\begin{proof}
The proof of positivity is directly based on the expression of the Hamiltonian constraint:
\begin{align}\label{ham-last}
 \frac{1}{2} e^{-4q + 2\sigma} (-\r^2 (\ptl_z A_\r - \ptl_\r A_z )^2 )  -e^{-2q-\sigma} \left( \halb \vert \grad_0 \sigma \vert^2 + 2 \Delta_0 \sigma + 2 \Delta_0 q -\frac{2}{\rho} \ptl_\rho q   \right)  \notag\\
  + tr(K)^2  - \Vert K \Vert^2_{\bar{q}} + 2 \Lambda \bar{g}(n, n) =0.
\end{align}
Integrating \eqref{ham-last} over a large coordinate ball $B(r)$ of radius $r$,

\begin{align*}
& \int_{B(r)} \Delta_0 \sigma + \Delta_0 q - \frac{1}{\r} \ptl_\r q \notag\\
&=\halb\int_{B(r)} \halb e^{-2q + 2\sigma} (\r^2 (\ptl_z A_\r - \ptl_\r A_z)^2) + e^{2q + \sigma} (-tr(K)^2 + \Vert \bar{K} \Vert^2_q - 2 \Lambda \bar{g}(n, n)) \notag\\
&\quad 
+ \halb  \int_{B(r)} \halb \vert \grad_0 \sigma \vert^2 
\end{align*}
The fall-off conditions for $q$ imply that the only boundary term
that remains is $ \lim_{r \to \infty} \int_{B(r)} \Delta_0 \sigma = \lim_{r \to \infty} \int_{\mathbb{S}^2(r)} \hat{n} (\grad_0 \sigma).$ Therefore, the limit $\lim_{r \to \infty} \int_{B(r)} \Delta_0 \sigma = -8 \pi^2 m_B (\bar{q}_{AF}, \bar{K}_{AF}).$ As a consequence, 
\begin{align}\label{mass-final}
m(\bar{q}, \bar{K}) = &\Omega \, m_B (\bar{q}_{AF}, \bar{K}_{AF})\notag\\
=& \frac{\Omega}{8 \pi}\int_{\mathbb{R}^3}  \left( \halb e^{-2q + 2\sigma} (\r^2 (\ptl_z A_\r - \ptl_\r A_z)^2) +  e^{2q + \sigma} ( \Vert \bar{K} \Vert^2_q) +  \halb \vert \grad_0 \sigma \vert^2 \right)\r d\r dz
\end{align}
with our time coordinate gauge condition $ \vert tr(K) \vert = \sqrt{2 \Lambda}$. It also follows from \eqref{mass-final} and the gauge condition   that $\bar{R}_{\bar{q}} >0$ and $m (\bar{q}, \bar{K}) >0.$ The convergence of the integral in \eqref{mass-final} follows from the regularity on the axis and fall-off conditions of the asymptotically flat $(\bar{\Sigma}_{AF}, \bar{q}_{AF}, \bar{K}_{AF})$ with convergent Brill mass. 
\end{proof}
We would like to point out some aspects of these formulas that could be useful to understand the open problems concerning the asymptotics, radiation and stability for Einstein's equations with $\Lambda >0$:

\begin{itemize}
\item The construction is robustly a large-data (strong field) construction as we make no smallness assumption of metric functions, angular-momentum or any other quantity.
\item An explicit formula for the geometric energy represented in terms of a dynamical metric quantities could be useful in studying the system from a PDE perspective.
\end{itemize}

However, a careful formulation of the Hamiltonian initial value problem is needed to study the evolution of such (mass-type) flux quantities at infinity (and at the horizons) for more general foliations and asymptotic conditions. A systematic study in this direction shall be postponed to a future work.
 %This additional structure, recasted in a Hamiltonian framework, has already been quite useful in \cite{GM17} for the $\Lambda =0$ case.
\subsection*{Acknowledgements}
The author gratefully acknowledges the encouraging interactions with Abhay Ashtekar in the early stages of this project. The author is also indebted to Vincent Moncrief for numerous encouraging discussions. A part of this work was done while the author was visiting the Institut des Hautes \'Etudes Scientifiques (IH\'ES) at Bures-sur-Yvette during the Fall 2016, where the author enjoyed the hospitality and conducive working conditions. 

\iffalse
\begin{verbatim}
Tasks:
0) PMT of course :)
1) Prove that the Liouvill's wave equation is satisied by dS and SdS, 
explicit representations of Kerr-deSitter
2)verify the KK reduction and give more details
3)Regularity conditions on the axis and the asymptotics
4) Discuss defocussing wave equation a little and mention the issues
\end{verbatim}
\fi
%%%%%%%%%%%%%%%%%%%%%%%%%%%%%%%%%%%%%%%%%%%%%%%%%%%%%%%%%%%%%%%%%%%%%%%%
 \bibliography{References_pmt.bib}
\bibliographystyle{amsplain}

\end{document}